\documentclass[11pt]{article}
\usepackage{caption}
\usepackage{graphicx}
\usepackage{amssymb}
\usepackage{amsmath}
\usepackage{hyperref}
\usepackage{thrmappendix}
\usepackage{afterpage}
\usepackage{color}
\usepackage{pseudocode}

\usepackage[default]{jasa_harvard}    
\usepackage{JASA_manu}

\newtheorem{theorem}{Theorem}[section]
\newtheorem{lemma}[theorem]{Lemma}

\newenvironment{proof}[1][Proof]{\begin{trivlist}
\item[\hskip \labelsep {\bfseries #1}]}{\end{trivlist}}
\newenvironment{definition}[1][Definition]{\begin{trivlist}
\item[\hskip \labelsep {\bfseries #1}]}{\end{trivlist}}

 \newcommand{\qed}{\hfill \ensuremath{\Box}}      
      
\begin{document}
\title{Bayesian hierarchical modeling of simply connected 2D shapes}
\author{Kelvin Gu, Debdeep Pati and David B. Dunson  \\
Department of Statistical Science, \\
 Duke University, NC 27708  \\
email: \texttt{gu.kelvin@gmail.com},  \texttt{debdeep.pati@stat.duke.edu}, \texttt{dunson@stat.duke.edu}
}
\maketitle

\small
\begin{center}
\textbf{Abstract}
\end{center}
Models for distributions of shapes contained within images can be widely used in biomedical applications
ranging from tumor tracking for targeted radiation therapy to classifying cells in a blood sample. Our focus is 
on hierarchical probability models for the shape and size of simply connected 2D closed curves, avoiding the 
need to specify landmarks through modeling the entire curve while borrowing information across curves for 
related objects.  Prevalent approaches follow a fundamentally different strategy in providing an initial point 
estimate of the curve and/or locations of landmarks, which are then fed into subsequent statistical analyses.  
Such two-stage methods ignore uncertainty in the first stage, and do not allow borrowing of information across 
objects in estimating object shapes and sizes.  Our fully Bayesian hierarchical model is based on 
multiscale deformations within a linear combination of cyclic basis characterization, which facilitates automatic alignment of the 
different curves accounting for uncertainty.  The characterization is shown to be highly flexible in representing 
2D closed curves, leading to a nonparametric Bayesian prior with large support.  Efficient Markov chain Monte
Carlo methods are developed for simultaneous analysis of many objects.  The methods are evaluated through
simulation examples and applied to yeast cell imaging data.
\vspace*{.3in}


\newpage

\section{Introduction}
Collections of shapes are widely studied across many disciplines, such as biomedical imaging, cytology and computer vision. Perhaps the most fundamental issue when studying shape is the choice of representation.  The simplest representations for shape are basic geometric objects, such as ellipses  \cite{cinquin1982hip,amenta1998new,rossi2003reconstruction}, polygons \cite{malladi1994evolutionary,malladi1995shape,sederberg19932,sato1997object}, and slightly more involved specifications such as superellipsoids \cite{gong2004parametric}.

Clearly, not all shapes can be adequately characterized by simple geometric objects. The landmark-based approach was developed to describe more complex shapes by reducing them to a finite set of landmark coordinates. This is appealing because the joint distribution of these landmarks is tractable to analyze, and because landmarks make registration/alignment of different shapes straightforward.  There is a very rich statistical literature on parametric joint distributions for multiple landmarks \cite{bookstein1986size,bookstein1996standard,bookstein1996shape,bookstein1996landmark,dryden1998statistical,dryden1993multivariate,mardia1989statistical,dryden2001surface,zheng2010automatic}, with some recent work on nonparametric distributions, both frequentist \cite{kume2007shape,kent2001functional,bhattacharya2008statistical,bhattacharya2009statistics} and Bayesian \cite{bhattacharya2010strong,bhattacharya2010nonparametric,bhattacharya2011nonparametric}.  

Unfortunately, in many applications it is not possible to define landmarks if the target collection of objects vary greatly. Furthermore, even if landmarks can be chosen, there may be substantial uncertainty in estimating their location, which is not accounted for in landmark-based statistical analyses.

In these situations, one can instead characterize shapes by describing their boundary, using a nonparametric curve (2D) or surface (3D). Curves and surfaces are widely used in biomedical imaging and commercial computer-aided design \cite{barnhill1985surfaces,lang1992developable,hagen1992variational,aziz2002bezier}, because they provide a flexible model for a broad range of objects e.g., cells, pollen grains, protein molecules, machine parts, etc.  
%
A collection of introductory work on curve and surface modeling can be found in \citeasnoun{su1989computational} and subsequent developments in \citeasnoun{muller2005surface}. Popular representations include: Bezier curves, splines, and principal curves \cite{hastie1989principal} (a nonlinear generalization of principal components, involving smooth curves which `pass through the middle' of a data cloud). \citeasnoun{anujcurve1} and \citeasnoun{anujcurve2} dealt with curve modeling based on smooth stochastic processes. Although there is a hugely vast literature on estimating curves
and surfaces, most of the focus is on estimating $\mu: \mathcal{X} \to \mathbb{R}$, where $\mathcal{X}$ a 
compact subset of $\mathbb{R}^p$ without making any constraints on $\mu$.  Estimating a closed surface or a curve involves a different modeling strategy and there has been very few works in this regime, particularly from a Bayesian point of view. To our knowledge, only \citeasnoun{pati2011surface} developed a Bayesian approach for fitting a closed surface using tensor-products.

Many of the above curve representations can successfully fit and describe complex shape boundaries, but they often have high or infinite dimensionality, and it is not clear how to directly analyze them. Also, they were not designed to facilitate comparison between shapes or characterize a collection of shapes.
One solution is to re-express each curve using Fourier descriptors or wavelet descriptors \cite{whitney1937regular,zahn1972fourier,mortenson1985geometrie,persoon1977shape}. Both approaches decompose a curve into components of different scales, so that the coarsest scale components carry the global approximation information while the finer scale components contain the local detailed information. Such multiscale transforms make it easier to compare objects that share the same coarse shape, but differ on finer details, or vice versa. The finer scale components can also be discarded to yield a finite and low-dimensional representation. Other dimensionality-reducing transformations include Principal Component Analysis and Distance Weighted Discrimination.


Note that the entire process is fragmented into three separate tasks: 1) curve fitting, 2) transformation, 3) population-level analysis. This can be problematic for several reasons. First, curve-fitting is not always accurate. If uncertainty is not accounted for, mistakes made during curve-fitting will be propagated into later analyses. Second, dimension-reducing transformations may throw away some of the information captured by curve-fitting. Finally, one suspects that the curve-fitting and transformation steps should be able to benefit from higher-level observations made during subsequent population analysis. For example, if the curve-fitting procedure is struggling to fit a missing or noisy shape boundary,  it should be able to draw on similar shapes in the population to achieve a more informed fit.
In this paper, we propose a Bayesian hierarchical model for 2D shapes, which addresses all of the aforementioned problems by performing curve fitting, multiscale transformation, and population analysis simultaneously within a single joint model.

The key innovation in our shape model is a shape-generating random process which can produce the whole range of simply-connected 2D shapes (shapes which contain no holes), by applying a sequence of multiscale deformations to a novel type of closed curve based on the work of \citeasnoun{róth2009cyclic}. \citeasnoun{mokhtarian1992theory}, \citeasnoun{dŽsidŽri2004multilevel} and \citeasnoun{dŽsidŽri2007nested} also proposed multiscale curves (with the latter two being more similar to our work, in their usage of B\'{e}zier curves and degree-elevation). However, none of these developed a statistical model around their representation or considered a collection of shapes.  In analyzing a population of shapes, a notion of average shape or mean shape is quite important.  \citeasnoun{dryden1998statistical} discussed notions of mean shape and shape variability and various methods of estimating them pertaining to landmark based analysis. We will follow a different but related strategy for defining the average shape in terms of the basis coefficients or the control points of the B\'{e}zier curves. We call it the `central shape'. Refer to \S \ref{ssec:rsp} for details.  To characterize shape variability, we also define a notion of shape quantile in \S \ref{ssec:rsp}.

In \S \ref{sec:rsp}, we describe the shape-generating random process, how it specifies a  multiscale probability distribution over shapes, and how this can be used to express various modeling assumptions, such as symmetry.  In \S \ref{sec:theory}, we provide theory regarding the flexibility of our model (support of the prior). In \S \ref{sec:modelprior} and \S \ref{sec:imagefit}, we show how the random process can be used to fit a curve to a point cloud or an image. In \S \ref{sec:pop}, we show how to simultaneously fit and characterize a collection of shapes, which also naturally incorporates inter-shape alignment. In \S \ref{sec:postcomp}, we describe the computational details of Bayesian inference behind each of the tasks described earlier. This results in a fast approximate algorithm which is scalable to a huge collection of shapes having a dense point cloud each. Finally, in \S \ref{sec:simstudy} and  \S \ref{sec:realstudy}, we test our model on simulated shapes and real image data respectively.
%
%
En route, we solve several important sub-problems that may be generally useful in the study of curve and surface fitting. First, we develop a model-based approach for parameterizing point cloud data. Second, we show how fully Bayesian joint modeling can be used to incorporate several pieces of auxiliary information in the process of curve-fitting, such as when each point within a point cloud also reports a surface orientation. Lastly, the concept of multi-scale deformation can be generalized to 3d surfaces in a straightforward manner.

\section{Priors for Multiscale Closed Curves}\label{sec:rsp}

\subsection{Overview}
Our random shape generation process starts with a closed curve and performs a sequence of multiscale deformations to generate a final shape. In \S \ref{ssec:rothcurve}, we introduce the Roth curve developed by \citeasnoun{róth2009cyclic}, which is used to represent the shape boundary. Then, in \S \ref{ssec:deform}, we demonstrate how to deform a Roth curve at multiple scales to produce any simply-connected shape. Using the mechanisms developed in \S \ref{ssec:rothcurve} and \S \ref{ssec:deform}, we present the full random shape process in \S \ref{ssec:rsp}.

\subsection{Roth curve}\label{ssec:rothcurve} 

A Roth curve is a closed parametric curve, $C: [-\pi,\pi] \rightarrow \mathbb{R}^2$, defined by a set of $2n+1$ control points $\{c_j, j=1,\ldots,2n+1\}$, where $n$ is the ÒdegreeÓ of the curve and we may choose it to be any positive integer, depending on how many control points are desired. For convenience, we will refer to the total number of control points as $J$, where $J(n) = 2n+1$. For notational simplicity, we will drop the dependence of $n$ in $J(n)$. As a function of $t$, the curve can be viewed as the trajectory of a particle over time. At every time $t$, the particle's location is defined as some convex combination of all control points. The weight accorded to each control point in this convex combination varies with time according to a set of basis functions, $\{B^n_j(t), j=1,\ldots,J\}$, where $B^n_j(t) > 0$ and $\sum_{j=1}^J B^n_j(t) = 1$ for all $t$.%
\begin{eqnarray}\label{eq:rothcurve}
C(t) &=& \sum_{j=1}^{J} c_j B_{j}^n(t), \enspace t \in [-\pi,\pi] \enspace, \\
B_{j}^n(t) &=& \frac{h_n}{2^n}\left\{1+ \cos\bigg(t + \frac{2\pi (j-1) } {2n+1}\bigg)\right\}^n, \enspace h_n = \frac{(2^n n!)^2}{(2n+1)!} \enspace,
\end{eqnarray}
where $c_j = [c_{j,x} \enspace c_{j,y}]'$ specifies the location of the $j^{th}$ control point and $B^n_j: [-\pi,\pi] \rightarrow [0,1]$ is the $j^{th}$ basis function. For simplicity, we omit the superscript $n$ denoting a basis function's degree, unless it requires special attention. This representation is a type of Bezier curve.  Refer to Figure \ref{fig:basis} for an illustration of the Roth basis functions. 
\begin{figure}
\begin{center}
\begin{tabular}{cc}
\includegraphics[width=0.55\textwidth]{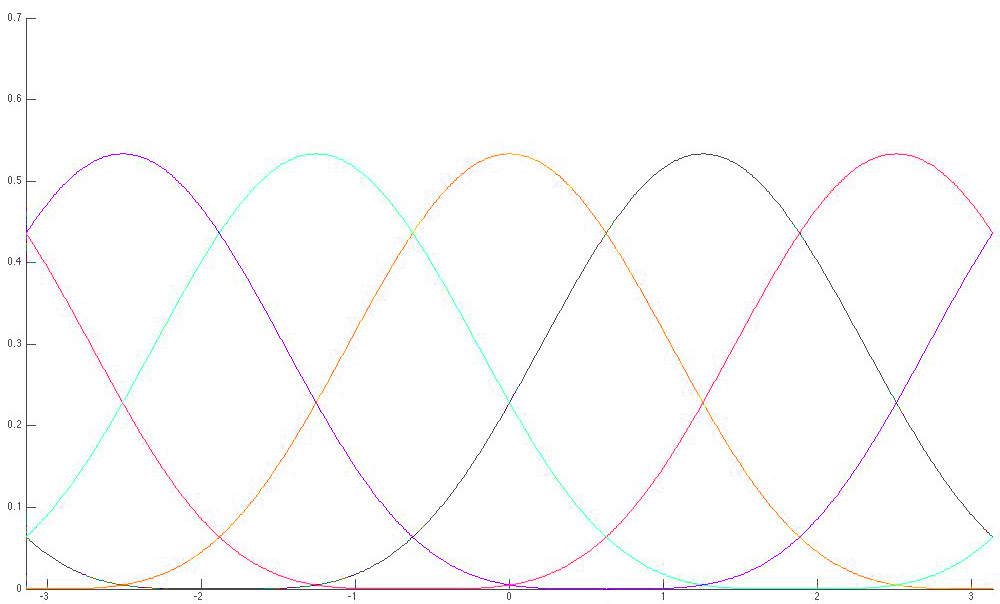}
\end{tabular}
\end{center} 
\caption{Roth basis}
\label{fig:basis}
\end{figure}
\noindent
The Roth curve has several appealing properties:
\begin{enumerate}
\item It is fully defined by a finite set of control points, despite being an infinite dimensional curve. 
\item It is always closed, i.e. $C(-\pi)=C(\pi)$. This is necessary to represent the boundary of a shape.
\item All basis functions are nonlinear translates of each other, and are evenly spaced over the interval $[-\pi,\pi]$. They can be cyclically permuted without altering the curve. This implies that each control point exerts the same `influence' over the curve. The influence of the control points is illustrated in \S \ref{ssec:influence}.
\item A degree 1 Roth curve having 3 control points is always a circle or ellipse.
\item Any closed curve can be approximated arbitrarily well by a Roth curve, for some large degree $n$. This is because the Roth basis, for a given $n$, spans the vector space of trigonometric polynomials of degree $n$ and as $n \rightarrow \infty$, the basis functions span the vector space of Fourier series. We elaborate on this in \S \ref{sec:theory}.
\item Roth curves are infinitely smooth in the sense that they are infinitely differentiable ($C^{\infty}$).
\end{enumerate}

\subsection{Deforming a Roth curve}\label{ssec:deform}

A Roth curve can be deformed simply by translating some of its control points. We now formally define deformation and illustrate it in Figure \ref{fig:deform}.
\begin{figure}
\begin{center}
\begin{tabular}{cc}
\includegraphics[width=0.5\textwidth]{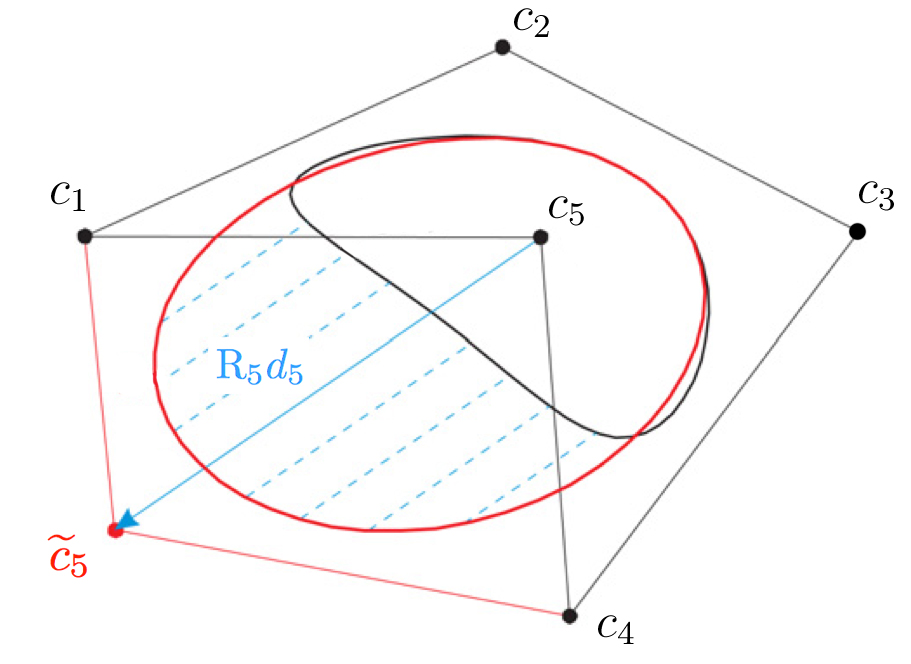}
\end{tabular}
\end{center} 
\caption{Deformation of a Roth curve}
\label{fig:deform}
\end{figure}
\begin{definition}
Suppose we are given two Roth curves,
\begin{eqnarray}
C(t) = \sum_{j=1}^{J}c_{j}B_{j}(t) , \quad
\widetilde{C}(t) = \sum_{j=1}^{J} \widetilde{c}_{j} B_{j}(t),
\end{eqnarray}
where for each $j$, $\widetilde{c}_{j} = c_j+ R_j d_j$, $d_j \in \mathbb{R}^2$ and $R_j$ is a rotation matrix. Then, we say that $C(t)$ is {\em deformed} into $\widetilde{C}(t)$ by the {\em deformation vectors} $\{d_j, j=1,\ldots,J\}$.
\end{definition}

Each $R_j$ orients the deformation vector $d_j$ relative to the original curve's surface. As a result, positive values for the y-component of $d_j$ always correspond to outward deformation, negative values always correspond to inward deformation, and $d_j$'s x-component corresponds to deformation parallel to the surface. We will call $R_j$ a {\it deformation-orienting matrix}. In precise terms,
\begin{eqnarray}\label{eq:rj}
R_j &=&
\begin{bmatrix}
\cos(\theta_j) & -\sin(\theta_j)\\ 
\sin(\theta_j) & \cos(\theta_j)
\end{bmatrix}
\end{eqnarray}
where $\theta_j$ is the angle of the curve's tangent line at $q_j = \frac{-2 \pi (j-1) }{2n + 1}$, the point where the control point $c_j$ has the strongest influence: $q_j = \arg\displaystyle\max_{t \in [a,b]} B_j(t)$. $\theta_j$ can be obtained by computing the first-derivative of the Roth curve, also known as its hodograph.
\begin{definition}
The hodograph of a Roth curve is given by:
\begin{eqnarray}\label{eq:hodo}
H(t) &=& \sum_{j=1}^{J}c_j\frac{d}{dt}B_{j}(t), \\
\frac{d}{dt}B_{j}(t) &=& -\frac{2}{(2n+1){2n \choose n}}\sum_{j=1}^{J} c_j \sum_{k=0}^{n-1}{2n \choose k} (n-k) \sin \bigg( (n-k)t+ 
\frac{2(n-k)(j-1)\pi}{2n+1}\bigg),
\end{eqnarray}
where $t \in [-\pi, \pi]$. If we view $C(t)$ as the trajectory of a particle, $H(t)$ intuitively gives the velocity of the particle at point $t$. 
\end{definition}
We can now use simple trigonometry to determine that
\begin{eqnarray}\label{eq:arctan}
\theta_j  = \arctan \left( \frac{H_y(q_j)}{H_x(q_j)} \right) .
\end{eqnarray}
Note that $R_j$ is ultimately just a function of $\{ c_{j} \in \mathbb{R}^2, j=1,\ldots, J \}$.

Next, we show how to alter the scale of deformation, using an important concept called degree elevation.
\begin{definition}
Given any Roth curve, we can use degree elevation to re-express the same curve using a larger number of control points (a higher degree). More precisely, if we are given a curve of degree $n$, $C(t) = \sum_{j=1}^{2n+1}c_{j}B_{j}^{n}(t)$, we can elevate its degree by any positive integer $v$, to obtain a new degree elevated curve: $\widehat{C}(t) = \sum_{j=1}^{2(n+v)+1}\widehat{c}_{j}B_{j}^{n+v}(t)$ such that $C(t)=\widehat{C}(t)$ for all $t \in [-\pi, \pi]$. In $\widehat{C}(t)$, each new degree-elevated control point, $\widehat{c}_{j}$, can be defined in terms of the original control points, $\{c_i, i=1,\ldots, 2n+1\}$:
\begin{eqnarray*}
\widehat{c}_j :=
\frac{1}{2n+1}\sum_{i=1}^{2n+1}c_i + \frac{{2(n+v)\choose n+v}h_n}{2^{2n-1}} \sum_{k=0}^{n-1} \frac{{2n\choose k}}{{2(n+v) \choose v+k}} \sum_{i=1}^{2n+1}\cos \left( (n-k)\left(\frac{-2(j-1)\pi}{2(n+v) +1}\right) + \frac{2(n-k)(i-1)\pi}{2n+1}\right) c_i \enspace .
\end{eqnarray*}
\end{definition}

It is crucial to note that the `influence' of a single control point shrinks after degree elevation. We quantify this intuition in \S \ref{ssec:influence}. This is because the curve is now shared by a greater total number of control points. This implies that after degree-elevation, the translation of any single control point will cause a smaller, finer-scale deformation to the curve's shape. Thus, degree elevation can be used to adjust the scale of deformation. We exploit this strategy in the random shape process proposed in \S \ref{ssec:rsp}.

To that end, we first rewrite all of the concepts described above in more compact vector notation. Note that the formulas for degree elevation, deformation, the hodograph and the curve itself all simply involve linear operations on the control points.

\subsection{Vector notation}
First, we rewrite the control points in a `stacked' vector of length $2J$.
\begin{eqnarray}
c= (c_{1,x},c_{1,y}, c_{2,x}, c_{2, y}, \ldots, c_{J,x}, c_{J,y})'. 
\end{eqnarray}
The formula for a Roth curve given in (\ref{eq:rothcurve}) can be rewritten as:
\begin{eqnarray}
C(t) &=& X(t) c \\
X(t) &=& \begin{bmatrix}
B_1(t) & 0 & B_2(t) & 0 & \cdots & B_{J}(t) & 0 \\
0 & B_1(t) & 0 & B_2(t) & \cdots & 0 & B_{J}(t)
\end{bmatrix}
\end{eqnarray}
The formula for the hodograph given in (\ref{eq:hodo}) is rewritten as:
\begin{eqnarray}\label{eq:xmatrix}
H(t) = \dot{X(t)} c, \quad 
\dot{X}(t) = \frac{d}{dt} X(t)
\end{eqnarray}
Deformation can be written as:
\begin{eqnarray}
\widetilde{c} = c + T(c) d, \quad
d = (d_{1,x},d_{1,y}, d_{2,x}, d_{2, y}, \ldots, d_{J,x}, d_{J,y})',  \quad T(c) = \mbox{block}(R_1, R_2, \ldots, R_J)
\end{eqnarray}
where $\mbox{block}(A_1, \ldots, A_q)$ is a $pq \times pq$ block diagonal matrix using $p\times p$ matrices $A_i, i=1, \ldots, q$. %
We call $T$ the {\it stacked} deformation-orientating matrix. Note that $T$ is a function of $c$, because each $R_j$ depends on $c$.
Degree elevation can be written as the linear operator, $E$:
\begin{eqnarray*}
\widehat{c} = E c, \quad E =  (E_{i,j})_{i=1, j=1}^{n+v,n}.
\end{eqnarray*}
where 
\begin{eqnarray*}
E_{i,j} &=& \frac{1}{2n+1} + \frac{{2(n+v)\choose n+v}h_n}{2^{2n-1}} \sum_{k=0}^{n-1} \frac{{2n\choose k}}{{2(n+v) \choose v+k}} \cos \left( (n-k)\left(\frac{-2(i-1)\pi}{2(n+v) +1}\right) + \frac{2(n-k)(j-1)\pi}{2n+1}\right). 
\end{eqnarray*}
We will maintain this vector notation throughout the rest of the paper.


\subsection{Random Shape Process}\label{ssec:rsp}
The random shape process starts with some initial Roth curve, specified by an initial set of control points, $c^{(0)}$.
From here on, we will refer to all curves by the stacked vector of
their control points, $c$. Then, drawing on the deformation and
degree-elevation operations defined earlier, we repeatedly apply the
following recursive operation $R$ times:
\begin{eqnarray}\label{eq:cprior}
\widehat{c}^{(r-1)}  =  E_{r}c^{(r-1)}, \quad d^{(r)} \sim  \mbox{N}(\mu_{r},\Sigma_{r}), \quad c^{(r)} =  \widehat{c}^{(r-1)}+T_{r}(c^{(r-1)})d^{(r)}
\end{eqnarray}
resulting in a final curve $c^{(R)}$. In other words, the process simply
has two steps: (i) degree elevate the current curve, (ii) randomly deform
it, and repeat a total of $R$ times. Note that this random process specifies a probability distribution over $c^{(R)}$.

\begin{figure}[h]
\centerline{\includegraphics[width=1\textwidth]{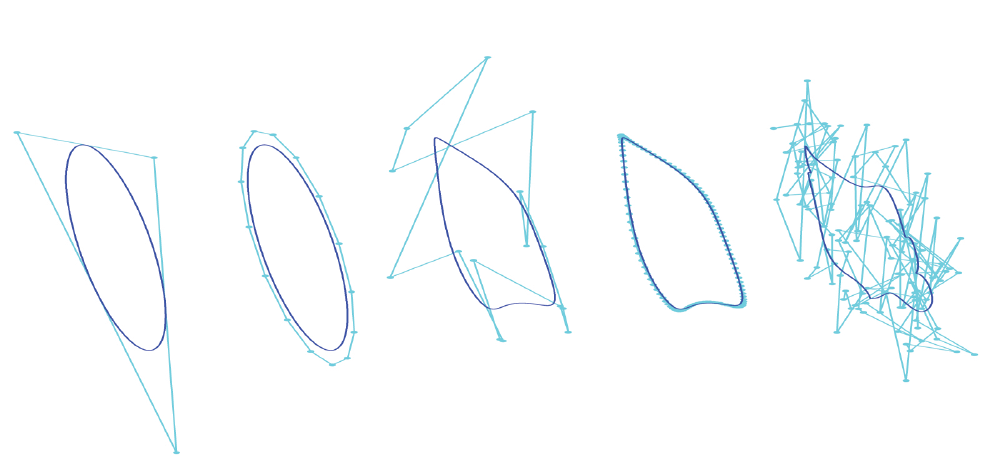}}
\caption{An illustration of the shape generation process. {\bf From left to right:} 1) initial curve specified by three control points, 2) the same curve after degree elevation, 3) deformation, 4) degree elevation again, 5) deformation again. Dark lines indicate the curve, pale dots indicate the curve's control points, and pale lines connect the control points in order. \label{fig:perturb} }
\end{figure} 

We now elaborate on the details of this recursive process. The parameters of the process are:
\begin{enumerate}
\item $R\in\mathbb{Z}$, the number of steps in the process
\item $n_{r}\in\mathbb{Z}$, the degree of the curve $c^{(r)}$, for each
$r=1,\ldots,R$. The sequence of $\{n_{r}\}_{1}^{R}$ must be strictly
monotonically increasing. For convenience, we will denote the number
of control points at a certain step $r$ to be $J_{r}=2n_{r}+1$.
\item $\mu_{r}\in\mathbb{R}^{2J_{r}}$, the average set of deformations applied
at step $r=1,\ldots,R$. Note that this vector contains a stack of
deformations, not just one.
\item $\Sigma_{r}\in\mathbb{R}^{2J_{r}\times2J_{r}}$, the covariance in
the set of deformations applied at step $r=1,\ldots,R$.
\end{enumerate}
According to these parameters, $E_{r}$ is then the degree-elevation
matrix going from degree $n_{r-1}$to $n_{r}$, $\mbox{N}(\cdot, \cdot)$ is a $2J_{r}$-variate
normal distribution and $T_{r}$ is the stacked deformation orienting
matrix.

We take special care in defining the initial curve, $c^{(0)}$. We
choose $c^{(0)}$ to be degree $n_{0}=1$, which guarantees that it
is an ellipse. For $j=1,2,3$, we define each control point as:  
\begin{eqnarray*}
c_{j}^{(0)} & = & (0,0)'+R_{\theta_{j}}d_{j}^{(0)},\\
R_{\theta_{j}} & = & \mbox{rotation matrix where }\theta_{j}=\frac{2\pi j}{3},
\end{eqnarray*}
and where each $d_{j}^{(0)}\in\mathbb{R}^{2}$ is a random deformation
vector. In words: we start with a curve that is just a point at the
origin, $C(t) \equiv (0,0)$, and apply three random deformations which are rotated by a radially symmetric amount: $0^\circ,120^\circ$ and $240^\circ$ (note that the final deformations are not radially symmetric, since each $d_j$ is randomly drawn). We will write this in vector notation as:
\begin{eqnarray*}
d^{(0)} & \sim & \mbox{N}(\mu_{0},\Sigma_{0})\\
c^{(0)} & = & {\bf 0}+T_{0}d^{(0)}
\end{eqnarray*}
The deformations essentially `inflate' the curve into some ellipse. This completes our definition of the random shape process.

We now give some intuition about the process and each of its parameters,
and define several additional concepts which make the process easier
to interpret. The random shape process gives a multiscale representation
of shapes, because each step in the process produces increasingly
fine-scale deformations, through degree-elevation.  

$R$ is then the number of scales or `resolutions' captured by the
process. Each $n_{r}$ specifies the number of control points at resolution
$r$. We will use $\mathbb{S}_{r}$ to denote the class of shapes
that can be exactly represented by a degree $n_{r}$ Roth curve.  See the definition of $\mathbb{H}^{n_r}$ in \S \ref{sec:theory} for 
a formal characterization of a special case of $\mathbb{S}_{r}$ . If $\{n_{r}\}_{1}^{R}$
is monotonically increasing, then $\mathbb{S}_{1}\subset \mathbb{S}_{2} \subset \ldots \subset \mathbb{S}_{R}$.
Thus, the deformations $d^{(r)}$ roughly describe the additional
details gained going from $\mathbb{S}_{r-1}$ to $\mathbb{S}_{r}$.

$\mu_{r}$ is the mean deformation at level $r$. Based on $\{\mu_{r},r=0,\ldots,R\}$,
we define the `central shape' of the random shape process, $c^{*}$ as:
\begin{eqnarray*}
c^{*} & := & c^{*(R)}\\
c^{*(r)} & = & E_{r}c^{*(r-1)}+T_{r}(c^{*(r-1)})\mu_{r}
\end{eqnarray*}

Note that $c^{*}$ is simply the deterministic result of the random
shape process when each $d^{(r)}=\mu_{r}$, rather than being drawn
from a distribution centered on $\mu_{r}$. Thus, all shapes generated by the process
tend to be deformed versions of the central shape. We illustrate this in Figure \ref{fig:moonstar}.
If the random shape process is used to describe a population of shapes,
the central shape provides a good summary.

\begin{figure}[h]
        \begin{tabular}{cc}
                \includegraphics[width=0.5\textwidth]{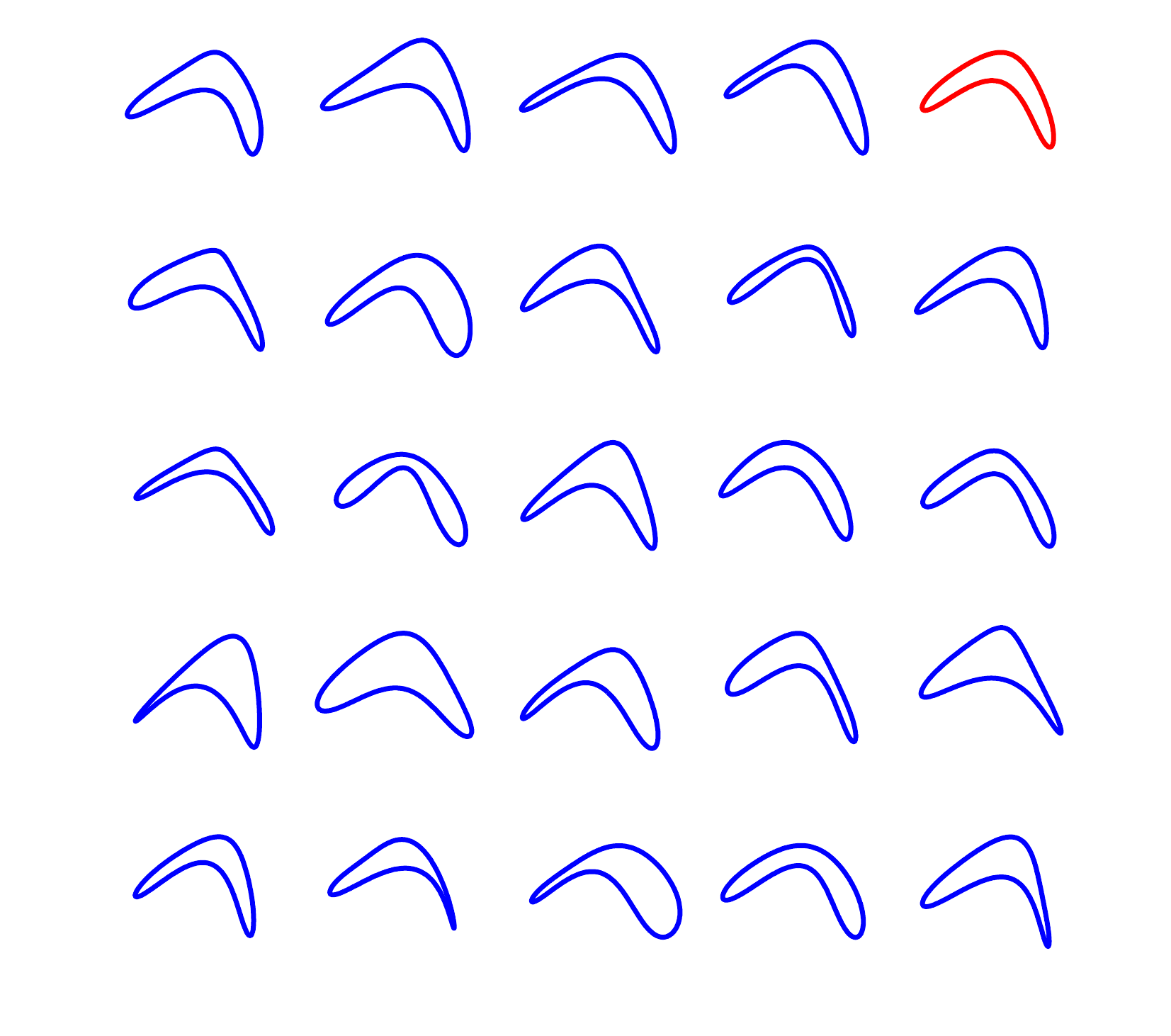} &  \includegraphics[width=0.5\textwidth]{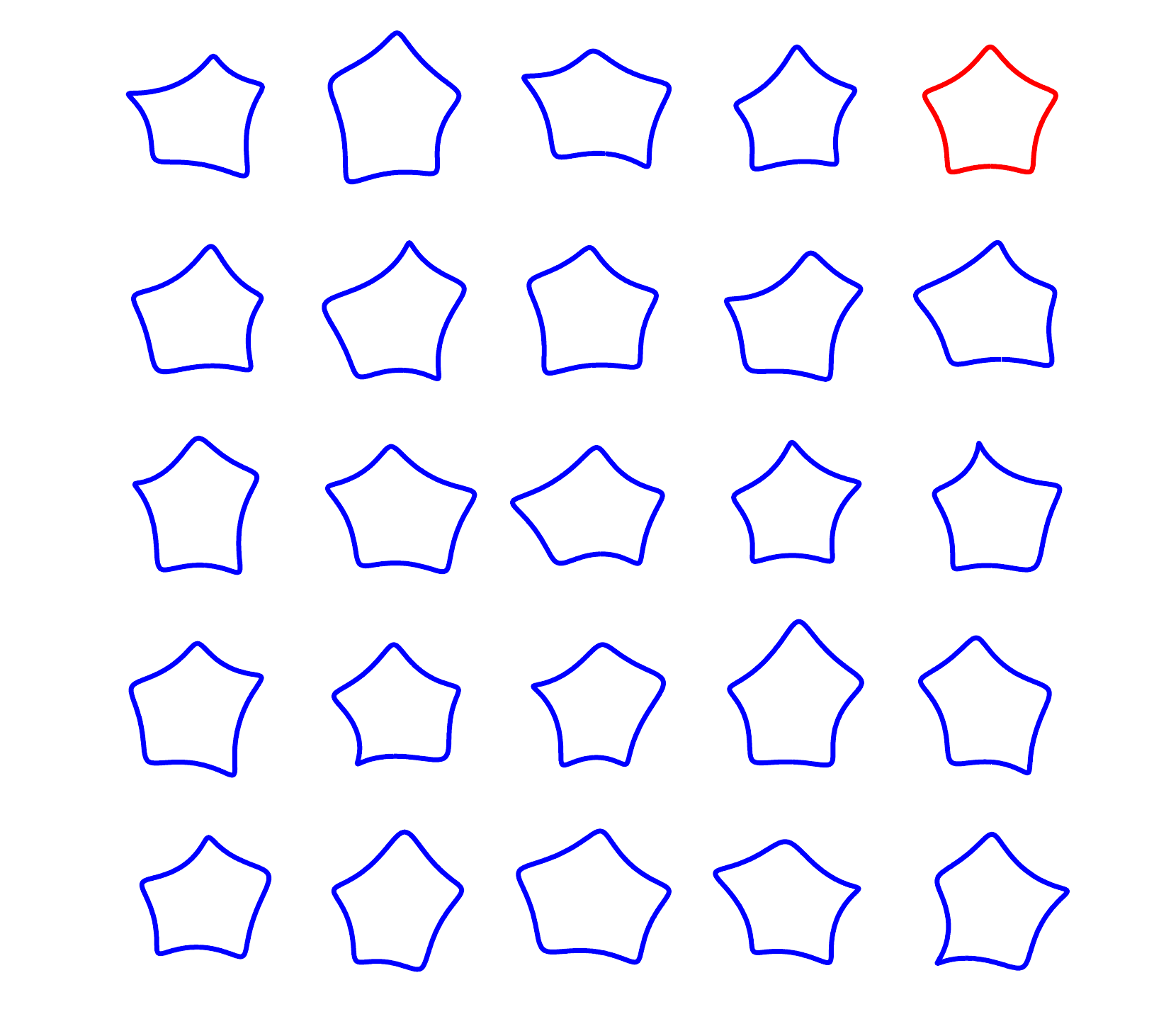}
        \end{tabular}
\caption{Realizations from the random shape process. The left panel shows realizations (in blue) when the central shape is a moon (shown in red). The right panel shows a similar case for a star.}
\label{fig:moonstar}
\end{figure}

$\Sigma_{r}$ determines the covariance of the deformations at level
$r$. This naturally controls the variability among shapes generated
by the process. If the variance is very small, all shapes will be
very similar to the central shape. $\Sigma_{r}$ can also be chosen to
induce correlation between deformation vectors at the same resolution,
in the typical way that correlation is induced between dimensions
of a multivariate normal distribution. This allows us to incorporate
higher-level assumptions about shape, such as reflected or radial
symmetry. For example, if $R=2,n_{1}=1$ and $n_{2}=2$, we can specify
perfect correlation in $\Sigma_{2}$, such that $d_{1}^{(2)}=d_{4}^{(2)}$
and $d_{2}^{(2)}=d_{3}^{(2)}$. The resulting shapes are guaranteed
to be symmetrical along an axis of reflection.

In the subsequent sections \ref{sec:modelprior} and \ref{sec:imagefit}, we show how to use our random shape process
to guide curve-fitting for various types of data. When doing so, we would like each resolution $r$ to describe the
shape as best as possible, within its class $\mathbb{S}_{r}$. This
can be achieved by setting $\mu_{r}={\bf 0}$ for $r=1,\ldots,R$
(not including $\mu_{0}$). 

%
%
%

\section{Properties of the Prior}\label{sec:theory}
\subsection{General notations}
 The supremum and $\mbox{L}_1$-norm are denoted by $||\cdot||_{\infty}$ and $||\cdot||_{1}$, respectively. We let $||\cdot||_{p, \nu}$ denote the norm of
$L_p(\nu)$, the space of measurable functions with $\nu$-integrable $p$th absolute power. The notation $C(\mathcal{X})$ is used for the space of continuous functions $f : \mathcal{X} \rightarrow \mathbb{R}$ endowed with the uniform norm. For $\alpha >0$ , we let $C^{\alpha}(\mathcal{X})$ denote the H\"{o}lder space of order $\alpha$, consisting of the functions $f \in C(\mathcal{X})$ that have $\lfloor \alpha \rfloor$ continuous
derivatives  with the $\lfloor \alpha \rfloor$th derivative
$f^{\lfloor \alpha \rfloor}$ being Lipshitz continuous of order $\alpha -\lfloor \alpha \rfloor$. 
We write ``$\precsim$'' for inequality up to a constant multiple and $\{a_{(1)}, a_{(2)}, \ldots, a_{(n)}\}$ to denote the order statistics of the set $\{a_i: a_i\in \mathbb{R} , i=1,\ldots,n\}$. 
\subsection{Support}
Let the H\"{o}lder class of periodic functions on $[-\pi, \pi]$ of order $\alpha$ be denoted by $C^{\alpha}([-\pi, \pi])$. Define the class of closed parametric curves $\mathcal{S}_{\mathcal{C}}(\alpha_1, \alpha_2)$ having different smoothness along different coordinates as
\begin{eqnarray}
\mathcal{S}_{\mathcal{C}}(\alpha_1, \alpha_2) := \{ S=(S^1, S^2): [-\pi, \pi] \to \mathbb{R} ^2, S^{i}  \in  C^{\alpha_i}([-\pi, \pi]), i=1,2\}.
\end{eqnarray}
Consider for simplicity a single resolution Roth curve with control points $\{c_j, j=0, \ldots, 2n\}$.   Assume we have independent Gaussian priors on each of the two coordinates of  $c_j$ for $j=0, \ldots, 2n$, i.e.,  $C(t) =\sum_{j=0}^{2n}c_j B_j^n(t), c_j \sim \mbox{N}_2(0, \sigma_j^2I_2), j=0, \ldots, 2n$.  Denote the prior for $C$ by $\Pi_{C^{n}}$.   $\Pi_{C^{n}}$ defines an independent Gaussian process for each of the components of $C$.  Technically speaking, the support of a prior is defined as the smallest closed set with probability one.  Intuitively, the support characterizes the variety of prior realizations along with those which are in their limit.  We construct a prior distribution to have large support so that the prior realizations are flexible enough to approximate the true underlying target object.  As reviewed in \citeasnoun{van2008reproducing},  the support of a Gaussian process (in our case $\Pi_{C^n}$) is the closure of the corresponding reproducing kernel Hilbert space (RKHS).  The following Lemma  \ref{lem:RKHS} describes the RKHS of $\Pi_{C^n}$,  which is a special case of Lemma 2 in \citeasnoun{pati2011surface}.

\begin{lemma}\label{lem:RKHS}
The RKHS $\mathbb{H}^{n}$ of $\Pi_{C^{n}}$ consists of all functions $h:[-\pi, \pi] \to \mathbb{R}^2$ of the form
\begin{eqnarray}
h(t)= \sum_{j=0}^{2n}c_{j}B_{j}^{n}(t)
\end{eqnarray}
where the weights $c_{j}$ range over $\mathbb{R} ^2$. The RKHS norm is given by
\begin{eqnarray}
||h||_{\mathbb{H}^{n}}^2 = \sum_{j=0}^{2n}||c_{j}||^2/ \sigma_j^2.
\end{eqnarray}
\end{lemma}

The following theorem describes how well an arbitrary closed parametric surface $S_0 \in \mathcal{S}_{\mathcal{C}}(\alpha_1, \alpha_2)$  can be approximated by the elements of $\mathbb{H}^{n}$ for each $n$.  Refer to Appendix \ref{app:results} for a proof. 

\begin{theorem}\label{thm:approx}
For any fixed $S_0 \in \mathcal{S}_{\mathcal{C}}(\alpha_1, \alpha_2)$, there exists $h \in \mathbb{H}^{n}$ with $||h||^2_{\mathbb{H}^{n}} \leq K_1\sum_{j=0}^{2n}1/\sigma_j^2$ such that
\begin{eqnarray}
|| S_0 - h||_{\infty} \leq K_2n^{-\alpha_{(1)}}\log n 
\end{eqnarray}
for some constants $K_1, K_2 > 0$ independent of $n$.
\end{theorem}
This shows that the Roth basis expansion is sufficiently flexible to approximate any closed curve arbitrarily well.  Although we have only shown large support of the prior under independent Gaussian priors on the control points, the multiscale structure should be even more flexible and hence rich enough to characterize any closed curve.  We can also expect minimax optimal posterior contraction rates using the prior $\Pi_{C^{n}}$ similar to Theorem 2 in \citeasnoun{pati2011surface} for suitable choices of prior distributions on $n$.

\subsection{Influence of the control points}\label{ssec:influence}
The unique maximum of basis function $B_j^n(t)$ defined in (\ref{eq:rothcurve}) is at $t =-2\pi (j-1) /J$, therefore the control point $c_j$ has the most significant effect on the shape of the curve in the neighborhood of the point $C(-2\pi (j-1)/J)$. Note that $B_j^n(t)$ vanishes at $t =  \pi - 2\pi(j-1)/J$, thus $c_j$ has no effect on the corresponding point i.e., the point of the curve is invariant under the modiÞcation of $c_j$.  The control point $c_j$ affects all other points of the curve, i.e. the curve is globally controlled. These properties are illustrated in Figure \ref{fig:influence}.

 \begin{figure}
 \begin{center}
\includegraphics[scale=0.3]{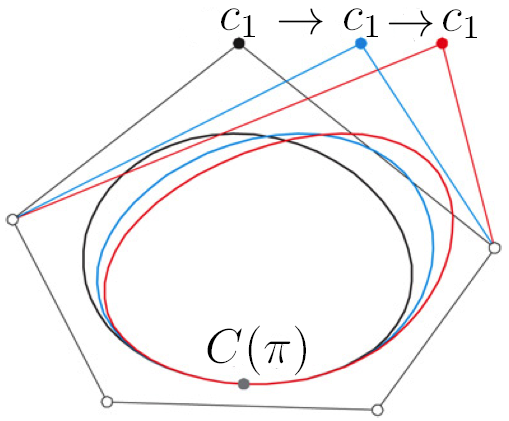}
\end{center}
\caption{Influence of the control point on the Roth curve}
\label{fig:influence}
\end{figure}

However, we emphasize following Proposition 5 in \citeasnoun{róth2009cyclic} that while control points have a global effect on the shape, this inßuence tends to be local and dramatically decreases on further parts of the curve, especially for higher values of $n$. 
%

\section{Inference from Point Cloud Data} \label{sec:modelprior}
We now demonstrate how our multiscale closed curve process can be used as a prior distribution for fitting a shape to a 2D point cloud. As a byproduct of fitting, we also obtain an intuitive description of the shape in terms of deformation vectors.
\noindent

Assume that the data consist of points $\{ p_i \in \Re^2, i=1,\ldots, N \}$ concentrated near a 2D closed curve.  Since a Roth curve can be thought of as a function expressing the trajectory of a particle over time, we view each data point, $p_i$, as a noisy observation of the particle's location at a given time $t_i$,
\begin{eqnarray}\label{eq:factor}
p_i = C(t_i) + \epsilon_i, \quad  \epsilon_i \sim N_2(0, \sigma^2 I_2) 
\end{eqnarray}
(\ref{eq:factor}) shares a similar form to nonlinear factor models, where $t_i$ is the latent factor score. %
We start by specifying the likelihood and prior distributions conditionally on the $t_i$s. We now rewrite the point cloud model in stacked vector notation. Defining
\begin{eqnarray*}
p =  (p_{1,x}, p_{1,y}, \ldots, p_{N,x}, p_{N,y} )',  &\quad& \epsilon = (\epsilon_{1,x}, \epsilon_{1,y}, \ldots, \epsilon_{N,x}, \epsilon_{N,y} )'  \\
t = (t_{1,x}, t_{1,y}, \ldots, t_{N,x}, t_{N,y})' , &\quad& \mathbb{X}(t)' =  [X(t_1)'  X(t_2)'  \ldots  X(t_N)']
\end{eqnarray*}
we have
\begin{eqnarray}\label{eq:modelmatrix}
p = \mathbb{X}(t) c + \epsilon, \quad
\epsilon \sim N_{2N}(0, \sigma^2 I_{2N})
\end{eqnarray}
where $X(t_i)$ is as defined in (\ref{eq:xmatrix}).

To fit a Roth curve through the data, we want to infer $P(c \mid p)$, the posterior distribution over control points $c$, given the data points $p$. To compute this, we must specify $P(p \mid c)$, the likelihood, and $P(c)$, the prior distribution over Roth curves specified by $c$. Refer to (\ref{eq:cprior}) in \S \ref{ssec:rsp} for a multiscale prior $P(c)$.  From (\ref{eq:modelmatrix}), we can specify the likelihood function as,
\begin{eqnarray}
P(\{p_i\}_1^N \mid\{c_i\}_1^J) &=& \displaystyle\prod_{i=1}^N N_{2}\bigg(p_i; \sum_{j=1}^J c_{j}B_{j}(t_i), \sigma^2 I_2\bigg), \quad 
P(p \mid c) = N_{2N}(p; \mathbb{X}(t) c,  \sigma^2 I_{2N}), \, \mbox{in vector notn.}
\end{eqnarray} 
This completes the Bayesian formulation for inferring $c$, given $p$ and $t$. In \S \ref{sec:postcomp}, we describe the exact method for performing Bayesian inference. As a byproduct of inference, we also infer the deformation vectors $d^{(r)}$ for , $r=1,\ldots,R$. Due to their multiscale organization, they may describe shape in a more intuitive manner than $\{c_j,j=1,\ldots,J\}$.

We propose a prior for $t_i$ conditionally on $c$, which is designed to be uniform over the curve's arc-length.  This prior is motivated by the frequentist literature on arc-length parameterizations \cite{madi2004closed}, but instead of replacing the points $\{ p_i \in \Re^2 \}$ with $\{ t_i \in [-\pi,\pi] \}$ in a deterministic preliminary step prior to statistical analysis, we use a Bayesian approach to formally accommodate uncertainty in parameterization of the points. 
Define the  arc-length function $A: [-\pi, \pi] \mapsto \mathbb{R}^+$ 
\begin{eqnarray}
A(u) := A(u; (c_0, \ldots, c_{2n})) = \int_{-\pi}^{u} ||H(t)||dt.
\end{eqnarray}
Note that $A$ is monotonically increasing and satisfies $A(-\pi) =0, A(\pi) = L(c_0, \ldots, c_{2n})$ where $L(c_0, \ldots, c_{2n})$ is the length of the curve conditional on the control points $(c_0, \ldots, c_{2n})$ and is given by $\int_{-\pi}^{\pi} ||H(t)||dt$.

Given $(c_0, \ldots, c_{2n})$, we draw $l_i \sim \mathrm{Unif}(0, L(c_0, \ldots, c_{2n}))$ and set $t_i = A^{-1}(l_i)$.  Thus we obtain a prior for the $t_i$'s which is uniform along the length of the curve.  We will discuss a novel griddy Gibbs algorithm for implementing the arc-length parametrization in a fully Bayesian framework in \S \ref{sec:postcomp}. 

\section{Inferences from Pixelated Image Data}\label{sec:imagefit}
In this section, we define a hierarchical Bayesian model for point cloud data concentrated near a 2d closed curve. We also show how image data  gives a bonus estimate for the object's surface orientation, $\omega_i$ at each point $p_i$. We incorporate this extra information into our model to obtain an even better shape fit, with essentially no sacrifice in computational efficiency.

A grayscale image can be treated as a function $Z: \mathbb{R}^2 \rightarrow  \mathbb{R}$. The gradient of this function, $\nabla Z:\mathbb{R}^2 \rightarrow \mathbb{R}^2$ is a vector field, where $\nabla Z(x,y)$ is a vector pointing in the direction of steepest ascent. In computer vision, it is well known that the gradient norm of the image, $||\nabla Z||_2: \mathbb{R}^2 \rightarrow  \mathbb{R}$ approximates a `line-drawing' of all the high-contrast edges in the image. Our goal is to fit the edges in the image with our shape model.

In practice, an image is discretized into pixels $\{z_{a,b} \mid a=1,\ldots,X,b=1,\ldots,Y\}$ but a discrete version of the gradient can still be computed by taking the difference between neighboring pixels, such that one gradient vector, $g_{a,b}$ is computed at each pixel. The image's gradient norm is then just another image, where each pixel $m_{a,b}=||g_{a,b}||_2$.

\noindent

Finally, we extract a point cloud: $\{(a,b) \mid m_{a,b} > M, a=1,\ldots,X,b=1,\ldots,Y\}$ where $M$ is some user-specified threshold. Each point $(a,b)$ can still be matched to a gradient vector $g_{a,b}$. For convenience, we will re-index them as $p_i$ and $g_i$. The gradient vector points in the direction of steepest change in contrast, i.e. it points across the edge of the object, approximating the object's surface normal. The surface orientation is then just $\omega_i = \arctan(\frac{g_{i,y}}{g_{i,x}})$.

In the following, we describe a model relating a Roth curve to each $\omega_i$. This model can be used together with the model we specified earlier for the $p_i$.

\subsection{Modeling surface orientation}
Denote by $v_i  = (H_x(t_i), H_y(t_i))\in \mathbb{R}^2$ the velocity vector of the curve $C(t)$ at the parameterization location $t_i, i =1, \ldots, N$. Note that $v_i$ is always tangent to the curve. Since each $\omega_i$ points roughly normal to the curve, we can rotate all of them by 90 degrees, $\theta_i = \omega_i + \frac{\pi}{2}$, and treat each $\theta_i$ as a noisy estimate of $v_i$'s orientation. Note that we cannot rotate the vector $g_i$ by 90 degrees and directly treat it as a noisy observation of $v_i$. In particular, $g_i$ 's magnitude bears no relationship to the magnitude of $v_i$: $||g_i||$ is the rate of change in image brightness when crossing the edge of the shape, while $||v_i||$ describes the speed at which the curve passes through $p_i$.

Suppose we did have some noisy observation of $v_i$, denoted $u_i$. Then, we could have specified the following linear model relating the curve $\{c_j,j=1,\ldots,J\}$ to the $u_i$'s:
\begin{eqnarray}\label{eq:surfacenormals}
u_i &=& v_i + \delta_i \\
&=& \sum_{j=1}^{J} c_j \frac{d}{dt}B_{j}(t_i) + \delta_i
\end{eqnarray}
for $i=1,\ldots, N$ where $\delta_i \sim N_{2}(0, \tau^2I_2)$. 
Instead, we only know the angle of $u_i$, $\theta_i$. In \S \ref{sec:postcomp}, we show that using this model, we can still write the likelihood for $\theta_i$, by marginalizing out the unknown magnitude of $u_i$. The resulting likelihood still results in conditional conjugacy of the control points.

\section{Fitting a population of shapes}\label{sec:pop}
We can easily generalize the methodology above to fit a collection of $K$ separate point clouds, and characterize the resulting population of shapes, represented by a closed curve. Continuing the vector notation earlier, we will represent the $k^{th}$ point cloud in a stacked vector $p^k$, the corresponding parametrizations $t^k$, surface orientations $\theta^k$, and the control points corresponding to that point cloud as $c^k$. Finally, we will denote the deformations which produce the curve for shape $k$ as $d^{(r),k}$ for each step $r$ in the shape process.

Up to this point, the parameters specifying each closed curve are separate and independent. This is sufficient if we just wish to fit each point cloud independently. However, now we aim to characterize all the curves as a single population. To do so, we treat each curve as an observation generated from a single random shape process.  We borrow information across the population of curves through sharing hyperparameters of our multiscale deformation model or by shrinking to a common value by assigning a hyperprior. These inferred hyperparameters and the uncertainty in estimating them  will effectively characterize the whole population of shapes. This is a hierarchical modeling strategy that is often used to characterize a population.

The hyperparameters of our random shape process are $\mu_r$ and $\Sigma_r$ for $r=1,\ldots,R$. We can treat each $\mu_r$ as an unknown and place the following prior on it:
\begin{eqnarray}
\mu_r &\sim& N_{2J_r}(\mu_{\mu_r}, \Sigma_{\mu_r})
\end{eqnarray}
By assuming that all shapes are generated from a single shape process with unknown $\mu_r$'s, we are basically assuming that all shapes are deformed versions of one `central shape' (defined earlier in \S \ref{ssec:rsp}) where the variability in deformation at scale $r$ is $\Sigma_r$.

Now suppose that each shape is rotated to a different angle. In this case, it may not be ideal to share {\it all} the deformations, because this would assume that all shapes are at exactly the same angle. One solution is to simply make the $\Sigma_1$ very large, allowing large variation in the $d^{(1),k}$. These deformations define the coarsest outline of each shape, which may be an ellipse. If these have wide variance, each coarse outline may be rotated to a different angle. Furthermore, since all subsequent deformations are defined {\it relative} to this initial coarse outline, different shapes can share the exact same deformations even if they are rotated to different angles.

Note that with this modification, all rotated shapes have essentially been aligned with each other, because all shape details expressed by the deformations for $r>1$ have been matched up.


\section{Posterior computation}\label{sec:postcomp} 
\subsection{An approximation to the deformation-orienting matrix for the deformation vector}
Observe that since $T_r(c^{(r),k})$ may not be linear in $c^{(r),k}$, due to the $\arctan$ in (\ref{eq:arctan}), the full conditional distribution of $c^{(r),k}$ is not conditionally conjugate. 
Below, we develop a novel approximation to the rotation matrix  $\tilde{T}_r$ of $T(c^{(r-1),k})$.  The approximation ensures that $T_r(c^{(r-1),k})$ is  linear in the level $r-1$ control points  $c^{(r-1),k}$ which results in conditional conjugacy of $c^{(r-1),k}$.  We resort to a Metropolis Hastings algorithm with an independent proposal suggested by the approximation to correct for the approximation error. For $j=1, \ldots, J_r$, let $q_{j,r}= -2\pi (j-1) / J_r$. 
 Recall that  $\dot{X}(q_{j,r}) c^{(r),k}$  is the hodograph evaluated at $q_{j,r}$. 
\begin{prop}
$R_j$ in (\ref{eq:rj}) can be  approximated by $\tilde{R}_j, j=1, \ldots, J_r$,  where $\tilde{R}_j$ are approximate rotation matrices given by 
\begin{eqnarray}
\frac{2\pi}{L_A} \times
\begin{bmatrix}
\dot{X}_{x}(q_{j,r}) c^{(r-1),k} & -\dot{X}_{y}(q_{j,r}) c^{(r-1),k}  \\ 
\dot{X}_{y}(q_{j,r}) c^{(r-1),k} & \dot{X}_{x}(q_{j,r})  c^{(r-1),k}
\end{bmatrix}
\end{eqnarray}
and $L_A$ is an approximation for the length of the curve $A(\pi; c^{(r-1),k})$ formed by the control points $c^{(r-1),k}$.
\end{prop}
\begin{proof}
Refer to the definition of $R_j$ in (\ref{eq:rj}).  First we derive an approximation for $\cos(\theta_j^*)$ and $\sin(\theta_j^*)$ where $\theta_j^*$ is the angle at the point $q_{j,r}$ of the curve specified by $c^{(r-1),k}$.  We write $\theta_j^*$ in vector notation 
\begin{eqnarray}
 \theta_j^*= \arctan\bigg(\frac{\dot{X}_{y}(q_{j,r}) c^{(r-1),k}}{\dot{X}_{x}(q_{j,r}) c^{(r-1),k}}\bigg).
\end{eqnarray}
Using the identities 
\begin{eqnarray}
\cos \arctan(x/y) =  \frac{x}{\sqrt{x^2 + y^2}}, \quad \sin \arctan(x/y) =  \frac{y}{\sqrt{x^2 + y^2}}, 
\end{eqnarray}
we obtain, 
\begin{eqnarray*}
\cos(\theta_j^*) &=& \frac{\dot{X}_x (q_{j,r})  c^{(r-1),k}}{\sqrt{(\dot{X}_{x}(q_{j,r}) c^{(r-1),k})^2 + (\dot{X}_{y}(q_{j,r})  c^{(r-1),k})^2}} \\
\sin(\theta_j^*) &=& \frac{\dot{X}_{y}(q_{j,r}) c^{(r-1),k}}{\sqrt{(\dot{X}_{x} (q_{j,r})  c^{(r-1),k})^2 + (\dot{X}_{y}(q_{j,r})  c^{(r-1),k})^2}} \end{eqnarray*}

The magnitude $\sqrt{(\dot{X}_{x}(q_{j,r}) c^{(r-1),k})^2 + (\dot{X}_{y}(q_{j,r}) c^{(r-1),k})^2}$ of the velocity vector at the point $q_{j,r}$ can be well approximated by the quantity $\frac{A(\pi; c^{(r-1),k})} {2 \pi}$ in view of the uniform arc-length parameterization discussed in \S  \ref{sec:para}.  Hence
\begin{eqnarray*}
\cos(\theta_j^*)  \approx \frac{2 \pi}{A(\pi; c^{(r-1),k})} \dot{X}_{x}(q_{j,r})  c^{(r-1),k} \\
\sin(\theta_j^*)   \approx  \frac{2 \pi}{A(\pi; c^{(r-1),k})} \dot{X}_{y}(q_{j,r})  c^{(r-1),k} 
\end{eqnarray*}
Plugging in in a fixed approximation $L_A$  for the length of the curve $A(\pi; c^{(1),k})$, we obtain the required result. 
\qed
\end{proof}

\subsection{Conditional posteriors for $m$ and $d^{(r)}$} 
Before deriving the conditional posteriors, we first introduce some
simplifying notation. Recall from \S \ref{ssec:rsp} that
\begin{eqnarray*}
c^{(r)} & = & E_{r}c^{(r-1)}+T_{r}\left(E_{r}c^{(r-1)}\right)d^{(r)}
\end{eqnarray*}
Using the approximation $\hat{T}_{r}c^{(r-1)}\approx T_{r}(E_{r}c^{(r-1)})d^{(r)}$,
we then have $c^{(r)}\approx\left(E_{r}+\hat{T}_{r}\right)c^{(r-1)}$.
In the new arrangement, we can now cleanly write $c^{(R)}$ in terms
of the base case, $c^{(0)}$.
\begin{eqnarray*}
\Omega_{a}^{b} & = & \begin{cases}
\prod_{r=a}^{b}\left(E_{r}+\hat{T}_{r}\right) & \quad \text{if}\, a < b\\
1 & \quad \text{otherwise.}
\end{cases}
\end{eqnarray*}
Hence $c^{(R)}$ can be approximated by
\begin{eqnarray*}
c^{(R)} & \approx & \Omega_{1}^{R}c^{(0)}.
\end{eqnarray*}
Note that the terms in $\Omega_{r+1}^{R}$ are the source of approximation
error. Given this expression, we can easily write $c^{(R)}$
in terms of $m$ and $d^{(0)}$,
\begin{eqnarray*}
c^{(R)} & \approx & \Omega_{1}^{R}\left(m+T_{0}d^{(0)}\right).
\end{eqnarray*}
We can also write $c^{(R)}$ in terms of $c^{(r-1)}$ and $d^{(r)}$,
for any $r=1,\ldots,R$
\begin{eqnarray*}
c^{(R)} & \approx & \Omega_{r+1}^{R}\left[E_{r}c^{(r-1)}+T_{r}\left(c^{(r-1)}\right)d^{(r)}\right]
\end{eqnarray*}
Note that as $r$ approaches $R$, $\Omega_{r+1}^{R}$ involves fewer
factors and the amount of approximation error decreases.

We are now ready to derive the conditional posteriors for $m^{k}$
and $d^{(r)}$ (as in \S \ref{sec:imagefit}, we are using a superscript $k$ to
denote variables for the $k^{\mathrm{th}}$ shape). First, we claim
that all posteriors can be written in the following form for generic `$x$', `$y$' and `$z$'.
\begin{eqnarray}
P(x\mid-) & \propto & \mbox{N}\left(y;Qx,\Sigma_{y}\right)\ \mbox{N}\left(x;z,\Sigma_{x}\right) \label{eq:post1}\\
P(x\mid-) & \sim & \mbox{N}\left(\hat{\mu},\hat{\Sigma}\right) \label{eq:post2}\\
\hat{\Sigma^{-1}} & = & \Sigma_{x}^{-1}+\sum_{k}Q'\Sigma_{y}^{-1}Q \nonumber\\
\hat{\mu} & = & \hat{\Sigma}\left(\Sigma_{x}^{-1}z+\sum_{k}Q'\Sigma_{y}^{-1}y\right) \nonumber.
\end{eqnarray}
Note that each conditional posterior is simply a multivariate normal.
We now prove that each posterior can be rearranged to match the form
of (\ref{eq:post1}) - (\ref{eq:post2}). 
\begin{eqnarray*}
P(m^{k}\mid-) & \propto & \mbox{N}\left(p^{k};\mathbb{X}(t^{k})c^{(R),k},\sigma^2 I_{2N^k}\right)\ \mbox{N}(m^{k};\mu_{m},\Sigma_{m})\\
 & \propto & \mbox{N}\left(p^{k};\mathbb{X}(t^{k})\Omega_{1}^{R,k}\left(m^{k}+T_{0}d^{(0),k}\right),\sigma^2 I_{2N^k}\right)\ \mbox{N}(m^{k};\mu_{m},\Sigma_{m})\\
 & \propto & \mbox{N}\left(p^{k}-\mathbb{X}(t^{k})\Omega_{1}^{R,k}T_{0}d^{(0),k};\mathbb{X}(t^{k})\Omega_{1}^{R,k}m^{k},\sigma^2 I_{2N^k}\right)\ \mbox{N}(m^{k};\mu_{m},\Sigma_{m})
\end{eqnarray*}
\begin{eqnarray*}
P(d^{(r)}\mid-) & \propto & \mbox{N}\left(p^{k};\mathbb{X}(t^{k})c^{(R),k},\sigma^2 I_{2N^k}\right)\ \mbox{N}(d^{(r)};\mu_{r},\Sigma_{r})\\
 & \propto & \mbox{N}\left(p^{k};\mathbb{X}(t^{k})\Omega_{r+1}^{R}\left[E_{r}c^{(r-1)}+T_{r}\left(c^{(r-1)}\right)d^{(r)}\right],\sigma^2 I_{2N^k}\right)\ \mbox{N}(d^{(r)};\mu_{r},\Sigma_{r})\\
 & \propto & \mbox{N}\left(p^{k}-\mathbb{X}(t^{k})\Omega_{r+1}^{R}E_{r}c^{(r-1)};\mathbb{X}(t^{k})\Omega_{r+1}^{R}T_{r}\left(c^{(r-1)}\right)d^{(r)},\sigma^2 I_{2N^k}\right)\ \mbox{N}(d^{(r)};\mu_{r},\Sigma_{r})
\end{eqnarray*}

\subsection{Conditional update for $\sigma^2 I_{2N^k}^2=\tau_{p}^{-1}$}
\begin{eqnarray*}
P(\tau_{p}\mid-) & \propto & \prod_{k=1}^{K}\prod_{i=1}^{N_{k}}\mbox{N}\left(p_{i}^{k};p_{i}^{k}*,\tau_{p}^{-1}\mathbb{I}_{2}\right)\ Ga\left(\tau_{p};\alpha,\beta\right)\\
 & \propto & \tau_{p}^{N_{tot}}\exp\left[-\frac{1}{2}\tau_{p}\sum_{k=1}^{K}\sum_{i=1}^{N_{k}}(p_{i}^{k}-p_{i}^{k}*)'(p_{i}^{k}-p_{i}^{k}*)\right]\ \tau_{p}^{\alpha-1}\exp\left(-\beta\tau_{p}\right)\\
 & \propto & \tau_{p}^{\alpha+N_{tot}-1}\exp\left[-\left(\beta+\frac{1}{2}\sum_{k=1}^{K}\sum_{i=1}^{N_{k}}\|p_{i}^{k}-p_{i}^{k}*\|_{2}^{2}\right)\tau_{p}\right].
\end{eqnarray*}
Where $p_{i}^{k}*=X(t_{i}^{k})c^{(R),k}$ and $N_{tot}=\sum_{k=1}^{K}N_{k}$.
The conditional posterior distribution is then
\begin{eqnarray*}
P(\tau_{p}\mid-) & \sim & Ga\left(\hat{\alpha},\hat{\beta}\right)
\end{eqnarray*}
where 
\begin{eqnarray*}
\hat{\alpha}=\alpha+N_{tot}-1\qquad\hat{\beta}=\beta+\frac{{\displaystyle \sum_{k=1}^{K}\sum_{i=1}^{N_{k}}}\|p_{i}^{k}-p_{i}^{k}*\|_{2}^{2}}{2}
\end{eqnarray*}
%

\subsection{Likelihood contribution from surface-normals} 
 Define
\begin{eqnarray}
\dot{X}_{x}(t_i) &=& \left[\frac{d B_0^{n_1}(t_i)}{dt}, 0,  \frac{d B_1^{n_1}(t_i)}{dt}, 0, \cdots,  \frac{d B_{2n_1}^{n_1}(t_i)}{dt},  0 \right]\\
\dot{X}_{y}(t_i) &=&  \left[0, \frac{d B_0^{n_1}(t_i)}{dt}, 0,  \frac{d B_1^{n_1}(t_i)}{dt}, \cdots,  0, \frac{d B_{2n_1}^{n_1}(t_i)}{dt} \right]
\end{eqnarray}
%
%
%
\begin{prop}
The likelihood contribution of the tangent directions $\theta_i^k, i=1, \ldots, N^k$ ensures conjugate updates of the control points for a multivariate normal prior. 
\end{prop}

\begin{proof}
Recall the noisy tangent director vectors $u_i^k$'s and $v_i^k$'s in (\ref{eq:surfacenormals}).   Use a simple reparameterization
\begin{eqnarray*}
u_i ^k= (m_i^k, m_i^k\tan\theta_i^k)
\end{eqnarray*}
where only $\theta_i'$s are observed and $m_i$'s aren't.  Observe that 
\begin{eqnarray}
v_i^k= (H_x(t_i), H_y(t_i)) =  (\dot{X}_{x}(t_i^k) c^{(3),k}, \dot{X}_{y}(t_i^k) c^{(3),k}).
\end{eqnarray}
Assuming a non-informative prior for the $m_i^k$'s on $\mathbb{R}$, the marginal likelihood of the tangent direction $\theta_i^k$ given $\tau^2$ and the parameterization $t_i^k$ is given by 
\begin{eqnarray*}
l(\theta_i^k) = \frac{1}{2\pi\tau^2}\int_{-\infty}^{\infty} \exp\left[-\frac{1}{2\tau^2}\{ (m_i^k - \dot{X}_{x}(t_i^k) c^{(3),k})^2 + (m_i^k\tan (\theta_i) - \dot{X}_{y}(t_i^k) c^{(3),k})^2 \}\right] dm_i^k
\end{eqnarray*}   
It turns out the above expression has a closed form given by 
\begin{eqnarray*}
 l(\theta_i^k) = \frac{1}{2\pi\tau^2}\frac{\sqrt{2\pi\tau^2}}{\sqrt{1+ \tan^{2}(\theta_i^k)}}\exp\left[  -\frac{1}{2\tau^2}\left\{(\dot{X}_{x}(t_i^k) c^{(3),k})^2 + (\dot{X}_{y}(t_i^k) c^{(3),k})^2 - \frac{(\dot{X}_{x}(t_i^k) c^{(3),k} + \dot{X}_{y}(t_i^k) c^{(3),k} \tan(\theta_i))^2}{1+\tan^2(\theta_i^k)}\right\} \right].  
\end{eqnarray*}
%
The likelihood for the $\{\theta_i^k, i=1, \ldots, N^k\}$ is given by 
\begin{eqnarray*}
L(\theta_1^k, \ldots, \theta^k_{N^k}) &\propto& \frac{1}{\tau^{N^k}}\exp\left[-\frac{1}{2\tau^2}\sum_{i=1}^{N^k}\frac{(\dot{X}_{x}(t_i^k) c^{(3),k})^2 \tan^2(\theta_i) + (\dot{X}_{y}(t_i^k) c^{(3),k})^2 -
2\dot{X}_{x}(t_i^k) c^{(1),k} \dot{X}_{y}(t_i^k) c^{(1),k}\tan(\theta_i)}{1+\tan^{2}(\theta_i)}\right]\\
&=& \frac{1}{\tau^{N^k}}\exp\left[-\frac{1}{2\tau^2} (c^{(3),k})'\left\{\sum_{i=1}^{N}(T_i^k)' \Sigma_i^kT_i^k \right\}c^{(3),k}\right]
\end{eqnarray*}
where  
\begin{eqnarray*}
\Sigma_i= \left( \begin{array}{cc}
\frac{\tan^{2}(\theta_i^k)}{1+\tan^{2}(\theta_i^k)} & \frac{-\tan(\theta_i^k)}{1+\tan^{2}(\theta_i^k)}  \\
\frac{-\tan(\theta_i^k)}{1+\tan^{2}(\theta_i^k)} & \frac{1}{1+\tan^{2}(\theta_i^k)}  \end{array} \right)
\end{eqnarray*}
and $T_i= [(\dot{X}_{x}(t_i^k) )' \quad (\dot{X}_{y}(t_i^k )']$ is a  $2(2n_3 +1) \times 2$ matrix. 
Clearly, an inverse-Gamma for $\tau^2$ and a multivariate normal prior for the control points are conjugate choices. \qed
\end{proof}

\subsection{Griddy Gibbs updates for the parameterizations $t_i^k$}
We discretize the possible values of $t_i^k \in [-\pi, \pi]$  to obtain a discrete approximation of its conditional posterior:
\begin{eqnarray*}
t_i^k \mid - \sim \frac{\mbox{N}(p_i^k;\mathbb{X}(t_i)'c^{(3),k}, \sigma^2 I_2)}{\sum_{\tau \in [-\pi, \pi]} \mbox{N}(p_i^k; \mathbb{X}(\tau)'c^{(3),k} ,\sigma^2 I_2)}
\end{eqnarray*}
We can make this arbitrarily accurate, by making a finer summation over $\tau$. 

\section{Simulation study}\label{sec:simstudy}

\section{Case study}\label{sec:realstudy}

\appendix
\section{Proofs of main  results}\label{app:results}
\noindent{Proof of Theorem \ref{thm:approx}:}
From \cite{stepanets1974approximation} and observing that the basis functions $\{B_{j}^n, j=0, \ldots, 2n\}$ span the vector space of trigonometric polynomials of degree at most $n$, it follows that given any $S_0^{i} \in C^{\alpha_i}([-\pi, \pi])$, there exists
$h^{i}(u) = \sum_{j=0}^{2n}c_{j}^{i}B_{j}^{n}(u)$, $h^{i}:[-\pi, \pi] \to \mathbb{R}$ with $|c_{j}^i| \leq M_i$, such that $||h^{i} - S_0^{i}||_{\infty} \leq K_in^{-\alpha_i}\log n$ for some constants $M_i, K_i > 0, i=1,2$.  Setting $h(u) = \sum_{j=0}^{2n}(c_{j}^{1}, c_{j}^{2})' B_{j}^{n}(u)$, we have
\begin{eqnarray*}
||h - S_0||_{\infty} \leq Mn^{-\alpha_{(1)}}\log n 
\end{eqnarray*}
with $||h||^{2}_{\mathbb{H}} \leq K \sum_{j=0}^{2n}\phi_{j}$
where $M=M_{(2)}, K=K_{(2)}$.

\noindent{Proof of Lemma \ref{lem:RKHS}:}

\bibliographystyle{ECA_jasa}
\bibliography{multicell_bib}

\end{document}